\documentclass[aps,pra,twocolumn,superscriptaddress,nofootinbib,
               longbibliography]{revtex4-2}
\usepackage{amsmath,amssymb,amsthm}
\usepackage{hyperref}
\usepackage{graphicx}
\usepackage{xcolor}
\usepackage{quantikz}
\usepackage{booktabs}
\usepackage{enumitem}

\newtheorem{theorem}{Theorem}
\newtheorem{lemma}[theorem]{Lemma}
\newtheorem{corollary}[theorem]{Corollary}
\newtheorem{proposition}[theorem]{Proposition}

\newtheorem{protocol}{Protocol}
\newtheorem{remark}{Remark}

\newcommand{\Wfree}{\mathcal{W}_{\mathrm{free}}}
\newcommand{\Stab}{\mathrm{Stab}}
\renewcommand{\ket}[1]{|{#1}\rangle}
\renewcommand{\bra}[1]{\langle{#1}|}
\providecommand{\tr}{\mathrm{tr}}
\renewcommand{\proj}[1]{|{#1}\rangle\langle{#1}|}

\begin{document}

\title{Magic Secret Sharing: Threshold Control of\\
Quantum Computational Power via GHZ Entanglement}

\author{Soumyojyoti Dutta}
\email{m24iqt014@iitj.ac.in}
\affiliation{Department of Physics, Indian Institute of Technology Jodhpur, 342030, India}

\author{Tushar}
\email{p24ph0012@iitj.ac.in}
\affiliation{Department of Physics, Indian Institute of Technology Jodhpur, 342030, India}

\date{August 3, 2026}

\begin{abstract}
We introduce \emph{Magic Secret Sharing} (MSS), a quantum cryptographic
primitive in which the secret is the \emph{computational capability}
of a quantum state rather than its classical description.
In the resource theory of magic, non-stabilizer states fuel universal
quantum computation via non-Clifford gates; MSS distributes this resource
with an $(n,n)$ threshold structure using a pre-shared GHZ state
and a single local phase gate $P(\varphi)=\mathrm{diag}(1,e^{i\varphi})$:
$n-1$ parties only measure and announce, and the remaining party receives the magic state.
Any individual party holds the maximally mixed state $I/2$,
with Wigner distance $C(I/2)=0$, so the share itself contributes no magic:
$I/2$ is left unchanged by every unitary, and any non-Clifford capability
a party ends up with must come from its own resources.
Any proper subset of parties holds a state independent of $\varphi$.
The recipient obtains $P(\varphi)\ket{+}$ exactly, with magic content
$C(\varphi)=(|\sin\varphi|+|\cos\varphi|-1)/2$,
enabling a logical $T$ gate via gate teleportation in
multi-server blind quantum computation (BQC).
The security condition has an exact geometric form: a dealer unitary $G$
keeps every share maximally mixed if and only if $G^\dagger\ket{+}$ lies on the
Bloch equator, which makes phase gates the diagonal solutions; this reflects the fact
that the GHZ code is perfectly secret exactly for equatorial states, where
$T$-type magic resides.
When the dealer varies her measurement setting, the recipient's assemblage is
steerable, and a witness built from the linear-programming dual of $C$
lower-bounds the magic of the recipient's conditional states from his
trusted measurements alone.
We demonstrate the three-party instance on \texttt{ibm\_marrakesh}
(156-qubit IBM Heron): security ($C(\rho_\mathrm{Bob}) < 10^{-6}$,
the linear-programming solver tolerance) holds in every run and is
independently reproduced on a second qubit assignment, and state
fidelity reaches $0.959$--$0.973$
for the authorised party, with faithfulness confirmed for all four
test values of $\varphi$ to within $0.025$ in magic content, a
residual that depolarising noise alone accounts for.
\end{abstract}

\maketitle

\section{Introduction}
\label{sec:intro}

Magic states are the scarce resource of fault-tolerant quantum
computation~\cite{Bravyi2005}: they fuel non-Clifford gates,
enabling computational universality beyond what is efficiently
classically simulable~\cite{Veitch2014,Howard2017}.
Yet the resource theory of magic has been studied almost exclusively
in the context of distillation and simulation overhead.
Its intersection with quantum cryptography remains largely unexplored.

Here we introduce \emph{Magic Secret Sharing} (MSS),
a cryptographic primitive that distributes computational power
rather than classical or quantum information.
MSS differs from standard quantum secret sharing
(QSS)~\cite{Hillery1999,Cleve1999} at the level of the threat model.
QSS prevents unauthorised parties from learning the \emph{identity}
of the secret state; MSS prevents them from \emph{non-Clifford computing with it}.
In the resource theory of magic, stabilizer operations are free
and non-stabilizer states are the costly resource.
An unauthorised party in MSS holds the maximally mixed state $I/2$,
which has zero magic content regardless of what operations
they apply to it --- free or otherwise, on any device.
The two guarantees are complementary: for proper subsets of parties the
protocol provides the standard secrecy guarantee of QSS (no information
about $\varphi$), and for single parties it additionally guarantees that the
share carries no magic at all.

The protocol uses a GHZ state and a single local phase gate
$P(\varphi) = \mathrm{diag}(1,e^{i\varphi})$
to achieve $(n,n)$ threshold MSS with magic content
$C(\varphi) = (|\sin\varphi|+|\cos\varphi|-1)/2$.
Equivalently, the GHZ state carries $P(\varphi)\ket{+}$ in the repetition code
$\ket{0}\mapsto\ket{0\cdots0}$, $\ket{1}\mapsto\ket{1\cdots1}$, and this code
hides a single-qubit state from every proper subset of shares exactly when
the state is equatorial. The security condition on the dealer's gate is the
same equator condition (Lemma~\ref{lem:column_sum}).
By measuring and announcing, the other $n-1$ parties deliver
$P(\varphi)|{+}\rangle$ to a server in distributed blind quantum
computation (BQC)~\cite{Fitzsimons2017}, enabling non-Clifford
computation while preventing any individual server from unilaterally
achieving it~\cite{GottesmanChuang1999}.
Mechanically, the three-party protocol is a remote state
preparation~\cite{Pati2000,Bennett2001} of an equatorial state whose
completion requires every party's announcement.

The delivery is mediated by quantum steering.
When the dealer's measurement setting varies, Charlie's conditional states
are distinct pure states, and no local hidden state (LHS) model can
reproduce them (Remark~\ref{rem:steering}).
Quantum steering has recently been identified as the required resource
for secure quantum \emph{state} sharing~\cite{Wilkinson2023};
in MSS the secret is computational capability
rather than state identity.
Separately, the linear-programming dual witnesses $H^*_x$ of the Wigner
distance $C$ provide a magic witness evaluated on Charlie's trusted
conditional states (Proposition~\ref{prop:1SDI}). It bounds their magic
content from below by $C(\varphi)$, tightly, without Alice revealing $\varphi$,
and at the $T$-gate point $\varphi=\pi/4$ it fixes the value exactly.
It complements Bell-inequality-based magic
witnesses~\cite{Macedo2025}, which are device-independent but certify
non-stabilizerness as a binary property; ours requires a trusted recipient
but bounds a specific magic measure.

We demonstrate the three-party protocol on \texttt{ibm\_marrakesh}
and achieve 4/4 faithfulness across test phases, with state fidelity
up to $0.973$ and security ($C(\rho_\mathrm{Bob}) < 10^{-6}$, the LP
solver tolerance) in every run; security was independently reproduced
on a second qubit assignment.
A companion paper~\cite{DuttaTushar2026b} studies non-Markovian
bath dynamics on the same hardware using the MSS circuit
as a diagnostic probe.

\section{Preliminaries}
\label{sec:prelim}

\subsection{Wigner distance and magic content}

We use the Wootters product construction~\cite{Wootters1987}.
For a single qubit, the four phase-point operators are
\begin{equation}
  A_{(q,p)} = \tfrac{1}{2}(I + (-1)^p X + (-1)^{q+p} Y + (-1)^q Z),
  \quad q,p \in \{0,1\},
\end{equation}
and the Wigner function is $W_\rho(q,p) = \frac{1}{2}\tr(\rho A_{(q,p)})$.
For $n$ qubits, $A_\alpha = A_{\alpha_1}\otimes\cdots\otimes A_{\alpha_n}$.

The \emph{stabilizer polytope} $\Wfree$ is the convex hull of
stabilizer state Wigner vectors.
The \emph{Wigner distance}~\cite{Dutta2026} is
\begin{equation}
  C(\rho) = \min_{f\in\Wfree}\|W_\rho - f\|_1.
  \label{eq:C_def}
\end{equation}
$C(\rho) = 0$ iff $\rho \in \Stab$; $C$ is invariant under
all single-qubit Clifford operations; $C(I/2) = 0$.
All security proofs in this work rely only on these two properties,
both of which follow directly from the LP definition~\eqref{eq:C_def};
extended properties of $C$ are developed in~\cite{Dutta2026}.

\subsection{Magic content of phase gate states}

\begin{lemma}[C-formula]
\label{lem:C_formula}
For $\ket{\varphi} = P(\varphi)\ket{+}$,
\begin{equation}
  C(\varphi) \equiv C(\proj{\varphi}) =
  \frac{|\sin\varphi| + |\cos\varphi| - 1}{2}.
  \label{eq:C_phi}
\end{equation}
\end{lemma}

\begin{proof}
The Bloch vector of $|{\varphi}\rangle$ is $(\cos\varphi,\sin\varphi,0)$, giving
$W_\varphi(q,p) = \frac{1}{4}(1+(-1)^p\cos\varphi+(-1)^{q+p}\sin\varphi)$.
For $\varphi\in(0,\pi/2)$ the unique negative entry is
$W_\varphi(0,1) = (1-\cos\varphi-\sin\varphi)/4 < 0$.
Since every $f\in\Wfree$ satisfies $f(q,p)\geq 0$,
\begin{equation}
  \|W_\varphi-f\|_1 \geq 2|W_\varphi(0,1)| = \frac{\cos\varphi+\sin\varphi-1}{2}.
\end{equation}
Equality is achieved at $f^* = \frac{\cos\varphi}{\cos\varphi+\sin\varphi}W_{|+\rangle}
+\frac{\sin\varphi}{\cos\varphi+\sin\varphi}W_{|+i\rangle}$.
Since $|+\rangle$ and $|+i\rangle$ are stabilizer states,
$W_{|+\rangle},W_{|+i\rangle}\in\Wfree$.
For $\varphi\in(0,\pi/2)$ both coefficients are positive and sum to unity,
so $f^*\in\Wfree$ is a valid convex combination, and
$C(\varphi) = (\cos\varphi+\sin\varphi-1)/2$.
Since $|\sin(\varphi+k\pi/2)|+|\cos(\varphi+k\pi/2)| = |\sin\varphi|+|\cos\varphi|$
for all integers $k$, the formula extends to all $\varphi$ by $\pi/2$-periodicity,
giving~\eqref{eq:C_phi}.
\end{proof}

$C(\varphi) = 0$ at $\varphi\in\{0,\pi/2,\pi,3\pi/2\}$.
Maximum $C = (\sqrt{2}-1)/2 \approx 0.207$ at $\varphi=\pi/4$ ($T$ gate).

\subsection{Phase gate security condition}

\begin{lemma}[Security condition: the equator]
\label{lem:column_sum}
Let $G$ be a single-qubit unitary applied by the dealer to qubit~$0$ of
$\ket{\mathrm{GHZ}_3}$, and let qubit~$0$ then be measured in the $X$ basis.
For either outcome, qubit~$1$'s reduced state is $I/2$ if and only if
\begin{equation}
  |G_{00}+G_{10}| = |G_{01}+G_{11}| = 1, \label{eq:colsum}
\end{equation}
equivalently, if and only if $G^\dagger\ket{+}$ lies on the equator of the
Bloch sphere. In that case the recipient, after the remaining $X$ measurement
and Pauli correction, holds the equatorial state
$\ket{0}+e^{i\delta}\ket{1}$ (normalised) with
$\delta=\arg[(G_{01}+G_{11})/(G_{00}+G_{10})]$, whose magic content is $C(\delta)$.
Every diagonal unitary $\mathrm{diag}(e^{ia},e^{ib})$ equals a phase gate
$P(b-a)$ up to a global phase and satisfies~\eqref{eq:colsum}.
\end{lemma}

\begin{proof}
Write $(a,b)=(G_{00}+G_{10},\,G_{01}+G_{11})=\sqrt2\,\bra{+}G$.
After $G$ and the projection of qubit~$0$ onto $\ket{+}$,
qubits~$1$,~$2$ hold
\begin{equation}
  \ket{\psi_{12}} \propto a\ket{00} + b\ket{11},
\end{equation}
so qubit~$1$'s populations are proportional to $|a|^2$ and $|b|^2$, and
$\rho_1=I/2$ if and only if $|a|=|b|$. Unitarity gives $|a|^2+|b|^2=2$, so this
is~\eqref{eq:colsum}. The vector $G^\dagger\ket{+}$ has components
$(\bar a,\bar b)/\sqrt2$, so $|a|=|b|$ says exactly that it is equatorial.
For the outcome $\ket{-}$, $G^\dagger\ket{-}$ is orthogonal to $G^\dagger\ket{+}$,
hence antipodal on the Bloch sphere and equatorial as well.
Measuring qubit~$1$ in the $X$ basis and correcting leaves qubit~$2$ in
$a\ket{0}+b\ket{1}\propto\ket{0}+e^{i\delta}\ket{1}$, whose magic is
$C(\delta)$ by Lemma~\ref{lem:C_formula}.
For $G=\mathrm{diag}(e^{ia},e^{ib})$, $|G_{00}|=|G_{11}|=1$ and $\delta=b-a$.
\end{proof}

\begin{remark}[Why the equator]
\label{rem:equator}
Encoding $\alpha\ket{0}+\beta\ket{1}$ in the repetition code
$\ket{0}\mapsto\ket{0\cdots0}$, $\ket{1}\mapsto\ket{1\cdots1}$ leaves every single
share in $\mathrm{diag}(|\alpha|^2,|\beta|^2)$, which reveals $|\alpha|^2$ unless
$|\alpha|=|\beta|$. The GHZ code therefore hides exactly the equatorial states,
and $T$-type magic lies on the equator. Lemma~\ref{lem:column_sum} is the same
statement seen from the dealer's side. Non-diagonal gates can also be secure
and faithful: $G=P(\varphi)X$ gives $\delta=-\varphi$ and $C(\delta)=C(\varphi)$,
while $G=e^{i\theta X}$ is secure but delivers $\ket{+}$, with $C=0$.
\end{remark}

\section{The Magic Secret Sharing Protocol}
\label{sec:protocol}

\subsection{Base case: three parties}

\begin{protocol}[MSS]
\label{def:protocol}
Let $\varphi\in(0,2\pi)\setminus\{0,\pi/2,\pi,3\pi/2\}$.
\begin{enumerate}[label=\textbf{Step \arabic*.}, leftmargin=4em]
  \item \textbf{GHZ.}
        Alice ($q_0$), Bob ($q_1$), Charlie ($q_2$) share
        $(\ket{000}+\ket{111})/\sqrt{2}$. Free (Clifford).
  \item \textbf{Inject.}
        Alice applies $P(\varphi)$:
        $\ket{\Psi_\varphi} =
        (\ket{0}\ket{00}+e^{i\varphi}\ket{1}\ket{11})/\sqrt{2}$.
  \item \textbf{Alice measures $X$,} broadcasts $m_A$.
        Bob--Charlie share
        $\ket{\psi_{BC}}\propto\ket{00}+(-1)^{[m_A=-]}e^{i\varphi}\ket{11}$.
  \item \textbf{Bob measures $X$,} broadcasts $m_B$.
  \item \textbf{Charlie} applies $Z^{[m_A=-]\oplus[m_B=-]}$, obtaining
        $P(\varphi)\ket{+}$ with $C = C(\varphi)$.
        (Equivalently, Charlie records this Pauli in his Pauli frame.)
\end{enumerate}
\end{protocol}

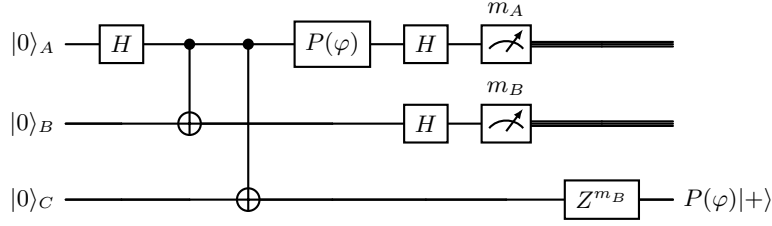
\begin{figure*}[!t]
\centering
\begin{quantikz}[column sep=0.45cm, row sep=0.5cm]
\lstick{$|0\rangle_A$} & \gate{H} & \ctrl{1} & \ctrl{2} & \gate{P(\varphi)} & \gate{H} & \meter{m_A} & \cw & \cwbend{2} \\
\lstick{$|0\rangle_B$} & \qw & \targ{} & \qw & \qw & \gate{H} & \meter{m_B} & \cwbend{1} & \wireoverride{n} \\
\lstick{$|0\rangle_C$} & \qw & \qw & \targ{} & \qw & \qw & \qw & \gate{Z^{m_B}} & \gate{Z^{m_A}} & \rstick{$P(\varphi)|{+}\rangle$}\qw
\end{quantikz}
\caption{Three-party MSS circuit.
Alice ($q_0$) prepares a GHZ state (H + two CX gates),
injects magic via $P(\varphi)$, and measures $X$, broadcasting $m_A$.
Bob ($q_1$) measures $X$ and broadcasts $m_B$.
Charlie ($q_2$) applies $Z^{m_A\oplus m_B}$ to obtain
$P(\varphi)|{+}\rangle$ with $C = C(\varphi)$.}
\label{fig:circuit}
\end{figure*}

\begin{theorem}[Three-party MSS]
\label{thm:23}
Protocol~\ref{def:protocol} achieves:
\textit{(i)} Faithfulness: $\rho_C = P(\varphi)\proj{+}P(\varphi)^\dagger$
in every measurement branch, so $C(\rho_C) = C(\varphi)$.
\textit{(ii)} Security: $\rho_B = I/2$, so $C(\rho_B) = 0$.
\textit{(iii)} Key indistinguishability:
$\rho_B(\varphi_1) = \rho_B(\varphi_2) = I/2$
$\forall$ valid $\varphi_1,\varphi_2$.
The circuit is shown in Fig.~\ref{fig:circuit}.
\end{theorem}

\begin{proof}
\textit{(i)} After Alice's and Bob's $X$ measurements
(outcomes $m_A,m_B\in\{+,-\}$, each branch with probability $1/4$),
Charlie's unnormalised state is
$|0\rangle_C + (-1)^{[m_A=-]\oplus[m_B=-]}e^{i\varphi}|1\rangle_C$,
since each $-$ outcome multiplies the relative phase by $-1$.
The correction $Z^{[m_A=-]\oplus[m_B=-]}$ maps every branch to
$|0\rangle_C+e^{i\varphi}|1\rangle_C \propto P(\varphi)|+\rangle$,
and Lemma~\ref{lem:C_formula} gives $C = C(\varphi)$.
We note that the value of $C$ alone could not certify the correction:
$C$ is invariant under Paulis, so $ZP(\varphi)\ket{+}$, which has fidelity
zero with the target, has the same $C$. Faithfulness is therefore stated
for the state itself.
\textit{(ii)} At Step~2, $\rho_B = I/2$ by no-signalling: Alice's local unitary
cannot change the reduced state of Bob and Charlie, which is that of
$\ket{\mathrm{GHZ}_3}$. After Step~3,
$|\psi_{BC}\rangle\propto|00\rangle+(-1)^{[m_A=-]}e^{i\varphi}|11\rangle$
is a (possibly phase-rotated) Bell state; either marginal is $I/2$, so $C(\rho_B)=0$.
\textit{(iii)} Since $\rho_B = I/2$ for all valid $\varphi$,
the trace distance $\frac{1}{2}\|\rho_B(\varphi_1)-\rho_B(\varphi_2)\|_1 = 0$.
\end{proof}

\begin{remark}[Noise-robust security: the $I/2$ fixed point]
Since $U\cdot(I/2)\cdot U^\dagger = I/2$ for any unitary $U$, and
$\mathcal{N}(I/2)=I/2$ for any unital channel $\mathcal{N}$,
Bob's share carries no magic resource under any unitary he applies and
any unital noise on his device. A party with its own non-stabilizer
resources has non-Clifford capability regardless; the statement is that
the share contributes none.
Under the standard magic-state injection model~\cite{Bravyi2005,GottesmanChuang1999},
in which non-stabilizer ancillas are the sole source of non-Clifford
universality, $C(\rho_B)=0$ is therefore sufficient to preclude
non-Clifford computational advantage. Security is noise-robust by construction.
\end{remark}

\begin{remark}[Adversary model]
\label{rem:adversary}
We assume throughout that a dishonest party may apply an arbitrary
CPTP map to their share, has unbounded classical computational power,
and is free to collude classically with other unauthorised parties.
All security statements below are made against adversaries of this
kind. The guarantee follows from Theorem~\ref{thm:23}(ii) and
Corollary~\ref{cor:coalition}: an
unauthorised share is $I/2$ whatever $\varphi$ happens to be, and no
operation of this kind can extract key-dependent information from a
state that does not depend on the key.
If a dishonest party additionally holds ancillas \emph{uncorrelated}
with the protocol state, their joint state is $I/2\otimes\rho_{\rm anc}$
and the protocol contributes no magic beyond what the ancillas already
carry: the protocol does not enhance the adversary's non-Clifford
resources. Security against ancillas \emph{correlated} with the
protocol state through side channels is not covered by the present
proofs; a full treatment would require decoupling-type arguments and is
left as an open problem.
\end{remark}

\begin{remark}[Steering]
\label{rem:steering}
Alice's phase gate followed by an $X$ measurement is a measurement of the
equatorial observable $P(\varphi)^\dagger X P(\varphi)=\cos\varphi\,X-\sin\varphi\,Y$
on the plain GHZ state, so the protocol is a remote preparation of Charlie's
state through the dealer's choice of measurement~\cite{WisemanJonesDoherty2007}.
For a single run with a fixed key, Alice uses one setting, and an assemblage
with one setting always admits a local hidden state (LHS) model (take the
hidden states to be the conditional states); steering is a property of a
family of settings. With two or more settings it is guaranteed here.
Consider the state $\tfrac{1}{\sqrt2}(\ket{00}+e^{i\varphi}\ket{11})$ shared by
Alice and Charlie after Bob's measurement and correction, with Alice measuring
$X$ or $Y$ (or, equally, two keys $\varphi_1\neq\varphi_2 \bmod \pi$).
Charlie's four normalised conditional states are distinct pure states.
In an LHS model, a pure conditional state can only be a mixture of hidden
states equal to itself, and each hidden state contributes to one outcome of
every setting; it would have to equal two distinct pure states, which is
impossible. The assemblage is therefore steerable for every valid $\varphi$,
with white-noise robustness $1/\sqrt2$, the same as a Bell state with two
Pauli settings. Note that an LHS model is not restricted to stabilizer hidden
states: an unsteerable assemblage may well contain magic states, so steering
is not implied by the presence of magic on Charlie's side alone.
In the hierarchy $\text{Entanglement} \supset \text{Steering}
\supset \text{Bell nonlocality}$~\cite{GallegoAolita2015},
the protocol uses steering; the Bell nonlocality of the GHZ state is not exploited.
\end{remark}

\begin{remark}[Minimal computational requirements on Charlie]
\label{rem:charlie_min}
The protocol places minimal demands on the recipient.
Charlie requires only:
\textit{(i)} one qubit of quantum memory,
\textit{(ii)} the ability to apply a single Pauli correction $Z$
(a free Clifford operation), and
\textit{(iii)} for the BQC application, the ability to perform
a CNOT, a computational-basis measurement and a Clifford correction
for gate teleportation.
No non-Clifford operations are required of Charlie at any stage ---
the magic resource is \emph{delivered} to him, not generated by him.
This is compatible with near-classical client models in distributed BQC.
\end{remark}

\subsection{\texorpdfstring{$n$}{n} parties: an \texorpdfstring{$(n,n)$}{(n,n)} scheme}

In the standard terminology of secret sharing, a $(k,n)$ threshold scheme lets
any $k$ shares reconstruct the secret while any $k-1$ shares reveal nothing.
In MSS every reconstruction uses all $n$ shares: $n-1$ of them are measured
and their outcomes announced, and the remaining share becomes the delivered
state. No $n-1$ shares can deliver the state to an outside party.

\begin{theorem}
\label{thm:nm1n}
The natural $n$-party extension is an $(n,n)$ MSS scheme: after any $n-1$
parties measure $X$ and announce their outcomes, the remaining party holds
$P(\varphi)\ket{+}$ exactly, and any proper subset of shares is independent of
$\varphi$ (Corollary~\ref{cor:coalition}).
\end{theorem}

\begin{proof}
We prove by induction on $j$ that after $j$ sequential $X$-measurements
(with outcomes $m_1,\ldots,m_j$ broadcast and corrected),
the remaining $n-j$ parties share
$|\psi^{(j)}\rangle = (|0^{n-j}\rangle + e^{i\varphi}|1^{n-j}\rangle)/\sqrt{2}$.

\emph{Base case} ($j=0$): this is the original GHZ state. \checkmark

\emph{Inductive step}: suppose $k \geq 2$ parties share
$|\psi\rangle = (|0^k\rangle + e^{i\varphi}|1^k\rangle)/\sqrt{2}$.
Party $k$ measures $X$ with outcome $m \in \{+,-\}$:
\begin{equation}
  \langle\pm|_k\,|\psi\rangle
  = \tfrac{1}{\sqrt{2}}\bigl(|0^{k-1}\rangle \pm e^{i\varphi}|1^{k-1}\rangle\bigr)/\sqrt{2}.
\end{equation}
Broadcasting $m$ allows the remaining $k-1$ parties to apply $Z^{[m=-]}$
to any one qubit, restoring
$|\psi^{(j+1)}\rangle = (|0^{k-1}\rangle + e^{i\varphi}|1^{k-1}\rangle)/\sqrt{2}$. \checkmark

\emph{Security at intermediate steps}: for $n-j \geq 2$,
any single qubit of $|\psi^{(j)}\rangle$ has reduced state
$\mathrm{tr}_{\text{rest}}(|\psi^{(j)}\rangle\langle\psi^{(j)}|) = I/2$,
so $C = 0$ for any individual party in the coalition.

\emph{Final step} ($j = n-1$): the last party holds
$|\psi^{(n-1)}\rangle = P(\varphi)|+\rangle$
with $C = C(\varphi)$ by Lemma~\ref{lem:C_formula}.
\end{proof}

\begin{corollary}[Coalition key indistinguishability]
\label{cor:coalition}
Let $S$ be any coalition, holding its quantum shares together with all
publicly broadcast measurement outcomes. As long as at least one party
outside $S$ has not yet measured, $\rho_S(\varphi_1) = \rho_S(\varphi_2)$
for all valid $\varphi_1, \varphi_2$.
\end{corollary}

\begin{proof}
By assumption at least one unmeasured qubit lies
outside $S$. Tracing it out destroys the coherence between
$|0\cdots0\rangle$ and $|1\cdots1\rangle$ in which the phase resides,
leaving (up to publicly known Pauli corrections)
$\rho_S = \tfrac{1}{2}\big(|0^{|S|}\rangle\langle 0^{|S|}| +
|1^{|S|}\rangle\langle 1^{|S|}|\big)$,
which is manifestly independent of $\varphi$.
The condition cannot be dropped: once every party outside $S$ has measured
and announced, $S$ holds a $\varphi$-dependent state; for $S$ containing only
the recipient, this is the intended delivery.
\end{proof}

\begin{remark}[Two-tier security structure]
\label{rem:twotier}
It is worth being precise about what is protected at each coalition
size, since the guarantee is not the same throughout.
A single party holds $I/2$. This is more than a stabilizer state: it is
\emph{absolutely stabilizer} in the sense of Zurel and
Davis~\cite{ZurelDavis2026}, sitting at the centre of the single-qubit
absolutely-stabilizer ball, so no unitary at all can turn that share
into a magic state.
Larger coalitions are in a weaker position. For $|S|\geq2$, while some party
outside $S$ has not yet measured,
the shared state is the rank-deficient mixture of
Corollary~\ref{cor:coalition}, with spectrum
$(1/2, 1/2, 0, \ldots, 0)$. It satisfies $C(\rho_S) = 0$, but it fails
the spectral conditions of~\cite{ZurelDavis2026}, so a coalition
\emph{can} generate magic from it using local non-Clifford operations.
What the coalition cannot do is generate magic that depends on
$\varphi$: by Corollary~\ref{cor:coalition} its state is the same for
every valid key, so whatever it produces carries no information about
the secret and gives no access to it. At the coalition level, then,
the guarantee is key indistinguishability rather than magic-freeness.
Whether a protocol state exists for which magic is absolutely
inaccessible to all sub-threshold coalitions is an open structural
question, and one we are currently pursuing.
\end{remark}

\section{Application to Blind Quantum Computation}
\label{sec:bqc}

\begin{theorem}[BQC gate control]
\label{thm:bqc}
In distributed BQC with $n$ servers:
\textit{(i)} Each server's share is $I/2$ and contributes no magic, so a server
with only Clifford operations cannot use it to implement a non-Clifford gate.
\textit{(ii)} When any $n-1$ servers measure and announce, the remaining server
holds $P(\varphi)\ket{+}$ and can apply $P(\pm\varphi)$ to a logical qubit by
gate teleportation~\cite{GottesmanChuang1999}; the teleportation correction
$P(2\varphi)$ is a Clifford gate, so that $P(\varphi)$ is implemented
deterministically, exactly when $\varphi\in\pi/4+(\pi/2)\mathbb{Z}$.
\textit{(iii)} In particular $\varphi=\pi/4$ implements the $T$ gate deterministically.
\end{theorem}

\begin{proof}
\textit{(i)} $C(U\cdot(I/2)\cdot U^\dagger) = 0$ for any $U$.
\textit{(ii)} By Theorem~\ref{thm:nm1n} the server holds $P(\varphi)\ket{+}$.
For a data qubit $\ket{\psi}$, a CNOT from data to ancilla followed by a
computational-basis measurement of the ancilla leaves $P(\varphi)\ket{\psi}$
(outcome $0$) or $P(-\varphi)\ket{\psi}$ (outcome $1$), the latter requiring the
correction $P(2\varphi)$. $P(\theta)$ is a Clifford gate if and only if
$\theta\in(\pi/2)\mathbb{Z}$. For other $\varphi$ the gate is obtained by
repeat-until-success with further resource states.
\textit{(iii)} $P(\pi/2)=S$.
\end{proof}

\begin{remark}[What can be certified with a trusted recipient]
\label{rem:certification}
With Charlie's device trusted and Alice's and Bob's devices untrusted, two
distinct statements can be certified. \textit{(a) Steering:} if Alice uses two
or more settings, a steering test on the full assemblage $\{\sigma_{a|x}\}$
(for instance, a steering SDP on Charlie's tomographic data, or test rounds in
which Alice announces her phase afterwards so that Charlie can align his
measurements) certifies entanglement without trusting the coalition's devices.
We do not perform such a test here. \textit{(b) Magic on the recipient's side:}
since Charlie's device is trusted, the magic of his conditional states can be
bounded from his own measurements, whatever the other devices do.
Proposition~\ref{prop:1SDI} gives a witness for (b); it does not certify (a).
Certification is destructive, so in practice a subset of runs is spent on it
and the remainder deliver undisturbed states for use.
Full device-independence via the Mermin inequality on the GHZ state
is a natural extension requiring additional protocol structure.
\end{remark}

\begin{proposition}[Magic witness on the recipient's conditional states]
\label{prop:1SDI}
Let Alice use Pauli measurement settings $x \in \{X, Y\}$ with outcomes
$a \in \{0,1\}$, and let $\tilde{\sigma}_{a|x}$ denote Charlie's normalised
conditional states after his $Z^{m_B}$ correction.
For $\theta\in\mathbb{R}$ let
\begin{equation}
  H^*(\theta) \;=\; \tfrac{1}{2}\bigl[\,\mathrm{sign}(\cos\theta)\,X
  \;+\; \mathrm{sign}(\sin\theta)\,Y\,\bigr]
  \label{eq:witness}
\end{equation}
denote the single-qubit LP dual witness of~\cite{Dutta2026} for the state with
Bloch vector $(\cos\theta,\sin\theta,0)$, and set
$H^*_X = H^*(\varphi)$ and $H^*_Y = H^*(\varphi-\pi/2)$.
Both satisfy
$\max_{\sigma\in\Stab}\tr(H^*_x\sigma) = \tfrac{1}{2} \equiv F_{\mathrm{stab}}$.
Define
\begin{equation}
  \mathcal{F}\bigl(\{\tilde{\sigma}_{a|x}\}\bigr)
  \;=\; \frac{1}{2}\sum_{x\in\{X,Y\}}\tr\bigl(H^*_x\,\tilde{\sigma}_{0|x}\bigr).
  \label{eq:steer_functional}
\end{equation}
Then:
\textit{(i)} $\mathcal{F}(\{\tilde{\sigma}_{a|x}\}) = C(\varphi) + F_{\mathrm{stab}}$.
\textit{(ii)} If Charlie's conditional states are stabilizer mixtures,
$\mathcal{F} \leq F_{\mathrm{stab}}$.
\textit{(iii)} The gap $\mathcal{F} - F_{\mathrm{stab}} = C(\varphi) > 0$
certifies, from Charlie's trusted measurements and without trusting
Alice's or Bob's devices, that $C(\tilde{\sigma}_{0|x}) \geq C(\varphi)$ for
$x\in\{X,Y\}$. The bound is tight, and equality is certified precisely
when $|\sin 2\varphi| = 1$ --- in particular at the $T$-gate point
$\varphi = \pi/4$ --- where the observed value saturates
$\tr(H^*_x\,\cdot\,)$ over the Bloch ball and fixes
$\tilde{\sigma}_{0|x}$ uniquely.
\end{proposition}

\begin{proof}
After Bob measures $X$ and Charlie applies $Z^{m_B}$, the Alice--Charlie state is
$\rho_{AC} = \frac{1}{2}(\ket{0}_A\ket{0}_C + e^{i\varphi}\ket{1}_A\ket{1}_C)(\mathrm{h.c.})$,
a pure entangled state for valid $\varphi$.

\emph{Conditional states.} For $x = X$, Alice's projector $\proj{+}$ gives
\begin{equation}
  \tilde{\sigma}_{0|X}
  = \frac{\bra{+}_A\,\rho_{AC}\,\ket{+}_A}{p(0|X)}
  = P(\varphi)\proj{+}P(\varphi)^\dagger ,
\end{equation}
with Bloch vector $(\cos\varphi,\sin\varphi,0)$.
For $x = Y$, since $\ket{+i} = (\ket{0}+i\ket{1})/\sqrt{2}$ and hence
$\bra{+i} = (\bra{0}-i\bra{1})/\sqrt{2}$,
\begin{equation}
  \tilde{\sigma}_{0|Y}
  \propto \ket{0}_C - ie^{i\varphi}\ket{1}_C
  = \ket{0}_C + e^{i(\varphi-\pi/2)}\ket{1}_C ,
\end{equation}
with Bloch vector $(\sin\varphi,-\cos\varphi,0)$.
By Lemma~\ref{lem:C_formula} and the $\pi/2$-periodicity of $C$,
$C(\tilde{\sigma}_{0|X}) = C(\tilde{\sigma}_{0|Y}) = C(\varphi)$.

\emph{Note}: for $x=Z$ (computational basis), Alice's projection gives
$\tilde{\sigma}_{0|Z} = \proj{0}_C$, a stabilizer state with $C=0$.
The functional therefore excludes the $Z$ setting, retaining only
$\{X,Y\}$ where magic is preserved under projection.

\textit{(i)} Write $\rho_\theta$ for the state with Bloch vector
$(\cos\theta,\sin\theta,0)$. The witness~\eqref{eq:witness} attains the LP
optimum at $\rho_\theta$:
\begin{equation}
  \tr\bigl(H^*(\theta)\,\rho_\theta\bigr) - F_{\mathrm{stab}}
  = \tfrac{1}{2}\bigl(|\cos\theta| + |\sin\theta| - 1\bigr)
  = C(\theta).
  \label{eq:opt_at_theta}
\end{equation}
The two conditional states sit at $\theta = \varphi$ and $\theta = \varphi-\pi/2$
respectively, and each is paired in~\eqref{eq:steer_functional} with its own
optimiser, so both terms contribute $C(\varphi)+F_{\mathrm{stab}}$ and
\begin{equation}
  \mathcal{F} = \tfrac{1}{2}\bigl[(C(\varphi)+F_{\mathrm{stab}})
  + (C(\varphi)+F_{\mathrm{stab}})\bigr]
  = C(\varphi)+F_{\mathrm{stab}} .
\end{equation}

\textit{(ii)} Both $H^*_X$ and $H^*_Y$ have the form
$\tfrac{1}{2}(\varepsilon_1 X + \varepsilon_2 Y)$ with
$\varepsilon_1,\varepsilon_2\in\{\pm1\}$. These four operators are Pauli
conjugates of one another (for instance $X(X{+}Y)X^\dagger = X{-}Y$), and
$\Stab$ is invariant under Clifford conjugation, so
$\max_{\sigma\in\Stab}\tr(H^*_x\sigma)$ takes the same value $1/2$ for both
settings; the common symbol $F_{\mathrm{stab}}$ is therefore well defined.
If $\tilde{\sigma}_{a|x} = \sum_\lambda p(\lambda|a,x)\rho_\lambda$
with $C(\rho_\lambda)=0$, then $\tr(H^*_x\rho_\lambda) \leq F_{\mathrm{stab}}$
for every $\lambda$ and every $x$. By linearity
$\tr(H^*_x\tilde{\sigma}_{a|x}) \leq F_{\mathrm{stab}}$, and averaging
over $x$ gives $\mathcal{F} \leq F_{\mathrm{stab}}$.

\textit{(iii)} The quantitative bound follows from the unconditional direction
of~\eqref{eq:opt_at_theta}: $H^*_x$ is dual-feasible for the linear
program~\eqref{eq:C_def} at every state, so
$\tr(H^*_x\rho) - F_{\mathrm{stab}} \leq C(\rho)$ for all $\rho$, and the
observed value $\tr(H^*_x\tilde{\sigma}_{0|x}) = C(\varphi)+F_{\mathrm{stab}}$
therefore gives $C(\tilde{\sigma}_{0|x}) \geq C(\varphi)$.
For the converse, the states consistent with that observation are those
whose Bloch vector lies in the disc cut from the Bloch ball by the plane
$\tr(H^*_x\rho) = C(\varphi)+F_{\mathrm{stab}}$, of radius
$\sqrt{(1-|\sin 2\varphi|)/2}$.
Since $H^*_x\in\mathrm{span}\{X,Y\}$, that plane contains the $Z$
direction, along which $C$ varies; the disc therefore carries states of
strictly larger magic unless its radius vanishes, which happens exactly
when $|\sin 2\varphi| = 1$.
Charlie's device is trusted, so measuring $H^*_x$ establishes this
without trusting Alice's or Bob's devices.
\end{proof}

\begin{remark}[$\mathcal{F}$ is not a steering functional]
\label{rem:notsteering}
The bound $F_{\mathrm{stab}}=\tfrac12$ restricts the hidden states to stabilizer
states, whereas a local hidden state model may use any states. An LHS model
that reproduces the observed marginals $p(0|x)=\tfrac12$ (half of the weight on
the strategy answering $0$ to $X$ only, with hidden state the top eigenvector of
$H^*_X$, and half on the strategy answering $0$ to $Y$ only, with the top
eigenvector of $H^*_Y$) attains $\mathcal{F}=1/\sqrt2=F_{\mathrm{stab}}+C(\pi/4)$,
which is the largest honest value at any $\varphi$; a semidefinite program confirms
that $1/\sqrt2$ is the maximum over all such LHS models. $\mathcal{F}$ is therefore
a magic witness on Charlie's conditional states, not a test of steering;
steering is established separately (Remark~\ref{rem:steering}).
\end{remark}

\begin{remark}[Why the witness must depend on the setting]
\label{rem:per_setting}
The functional~\eqref{eq:steer_functional} pairs each conditional state with
its own witness, and this is not cosmetic. Equation~\eqref{eq:opt_at_theta}
is an identity only at the state for which $H^*(\theta)$ is the LP optimiser;
away from it one has the inequality
$\tr(H^*\rho) - F_{\mathrm{stab}} \leq C(\rho)$.
No fixed $H^*$ can turn this into an identity for all $\rho$, since $C$ is the
value function of the linear program~\eqref{eq:C_def} and hence convex and
piecewise linear, whereas $\rho\mapsto\tr(H^*\rho)$ is linear.
The optimiser~\eqref{eq:witness} is piecewise constant in the quadrant of
$\theta$~\cite{Dutta2026}, and the two settings sit exactly one quadrant apart
--- $\pi/2$ being precisely the period of $C$ --- so they never share an
optimiser: $H^*_X = \tfrac12(X{+}Y)$ while $H^*_Y = \tfrac12(X{-}Y)$ for
$\varphi\in(0,\pi/2)$. Using a single witness on both branches makes the second
term negative and the gap vanish; indeed
$\max_{H}\min_x[\tr(H\tilde\sigma_{0|x}) - \max_{\sigma\in\Stab}\tr(H\sigma)] = 0$,
so no single witness certifies anything in this two-setting scheme.
\end{remark}

\begin{remark}[Scope of the witness]
\label{rem:scope}
For the honest assemblage the violation equals $C(\varphi)$ exactly, so
the certification is quantitative rather than binary: it establishes
$C(\tilde{\sigma}_{0|x}) \geq C(\varphi)$ without Alice revealing
$\varphi$, and the bound is tight.
It does not in general \emph{determine} $C(\tilde{\sigma}_{0|x})$.
Each $H^*_x$ lies in $\mathrm{span}\{X,Y\}$ and is therefore blind to
the $Z$ component of Charlie's Bloch vector, while $C$ is not.
Given the observed value, the normalised conditional states consistent
with it form a disc of radius $\sqrt{(1-|\sin 2\varphi|)/2}$ in the
Bloch ball; no-signalling and the observed marginals impose no further
restriction, since with $p(0|x)=1/2$ the choice
$2m = (n_{0|X}+n_{0|Y})/2$ yields admissible complementary elements for
any pair of conditional states.
The disc degenerates to a point precisely when $|\sin 2\varphi| = 1$
--- at $\varphi = \pi/4$ and its $\pi/2$-translates, including the
$T$-gate case of Theorem~\ref{thm:bqc}(iii) --- where equality is
certified. Elsewhere it is not: at $\varphi = \pi/6$ the same
statistics admit $C = 0.3124$ against $C(\varphi) = 0.1830$.
Determining the value requires full tomography on Charlie's side,
which the experiment performs but which goes beyond the witness.
\end{remark}

\section{IBM Quantum Hardware Experiment}
\label{sec:hardware}

\subsection{Setup}

The three-party MSS protocol is implemented on \texttt{ibm\_marrakesh},
a 156-qubit IBM Heron r2 processor (\texttt{measure\_all()};
4096 shots per circuit).
The three roles are pinned to physical qubits $(q_2, q_1, q_3)$ for
(Alice, Bob, Charlie): the protocol requires Alice to be coupled to
both other parties, and this triple satisfies that directly, so the
transpiled circuits contain no routing (verified: zero \textsc{swap}
gates, two \textsc{cz} gates, depth $\leq 19$).
At the time of the run these qubits had
$T_1 = 293.5, 200.8, 278.8\,\mu$s and $T_2 = 297.4, 151.6, 197.9\,\mu$s
respectively.
State tomography uses $X/Y/Z$ measurements on Charlie (Bob).
Post-selection on Alice's $\ket{+}$ outcome and software $Z$-correction
on Charlie are applied classically in the Qiskit LSb-0 convention.
Post-selection is a choice of analysis, not a requirement of the protocol:
correcting with $Z^{[m_A=-]\oplus[m_B=-]}$ retains every shot.
$C(\rho)$ is computed via linear programming against the six
single-qubit stabilizer Wigner vectors.\footnote{IBM Quantum job ID (\texttt{ibm\_marrakesh}, standard protocol): \texttt{d9o1f3oqs0bc73e3vo1g}.}

\subsection{Standard protocol}

\begin{table*}[!t]
\centering
\caption{Standard MSS on \texttt{ibm\_marrakesh} (4096 shots, qubits $(2,1,3)$).
$C_\mathrm{th} = C(\varphi)$ from Lemma~\ref{lem:C_formula}.
$\sigma_C$ and $\sigma_F$ are $1\sigma$ parametric bootstrap uncertainties
($N_\mathrm{boot}=2000$, $N_\mathrm{eff}\approx 2030$ post-selected shots,
giving $\sigma \approx 0.022$ per Pauli estimator).
$C(\rho_B) < 10^{-6}$ in all cases, set by the LP solver tolerance rather
than by statistics; see the discussion following the table.}
\label{tab:hw}
\begin{tabular}{ccccccc}
\toprule
$\varphi$ & $C_\mathrm{th}$ & $C(\rho_C)$ & $\sigma_C$ & Fidelity & $\sigma_F$ & $C(\rho_B)$ \\
\midrule
$\pi/8$   & 0.153 & 0.128 & 0.016 & 0.973 & 0.006 & $<10^{-6}$ \\
$\pi/4$   & 0.207 & 0.194 & 0.016 & 0.959 & 0.008 & $<10^{-6}$ \\
$\pi/3$   & 0.183 & 0.159 & 0.016 & 0.965 & 0.008 & $<10^{-6}$ \\
$3\pi/4$  & 0.207 & 0.188 & 0.016 & 0.959 & 0.009 & $<10^{-6}$ \\
\bottomrule
\end{tabular}
\end{table*}

\begin{figure*}[!tp]
\centering
\includegraphics[width=\textwidth]{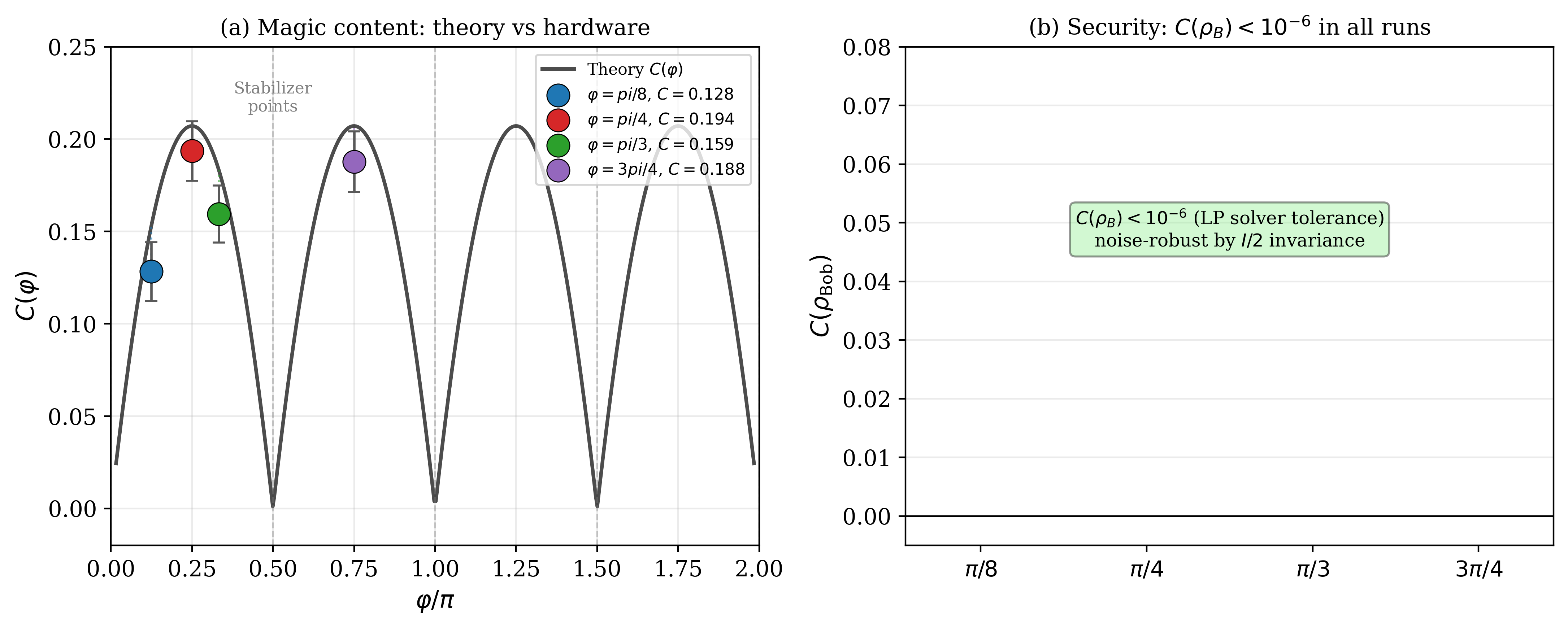}
\caption{Standard MSS on \texttt{ibm\_marrakesh} (4096 shots, qubits $(2,1,3)$).
\textbf{(a)} Measured $C(\rho_C)$ (scatter, $1\sigma$ bootstrap error bars)
against the theory curve $C(\varphi) = (|\sin\varphi|+|\cos\varphi|-1)/2$ (line).
All four $\varphi$ values are faithful to within $0.025$.
\textbf{(b)} $C(\rho_B) < 10^{-6}$ in every case ---
noise-robust security from the $I/2$ invariance argument.}
\label{fig:standard}
\end{figure*}

Table~\ref{tab:hw} summarises the measured magic content, fidelity,
and bootstrap uncertainties for all four test phases; Fig.~\ref{fig:standard}
displays the same data against the analytic curve of Lemma~\ref{lem:C_formula}.
Post-selection on Alice's $\ket{+}$ outcome retains approximately half
the shots, giving $N_{\rm eff}\approx 2030$ effective samples per
tomographic setting. Uncertainties are obtained by parametric bootstrap:
Pauli expectation values are resampled from the binomial shot
distribution, the Bloch vector is reconstructed and clipped to the
physical ball, and $C$ is recomputed by linear programming for each of
$N_{\rm boot}=2000$ replicas.

$C(\rho_B) < 10^{-6}$ in all cases, consistent with the noise-robust
prediction of the $I/2$ invariance argument. This bound is set by the
solver tolerance rather than by statistics: the stabilizer polytope
boundary requires a Bloch component sum exceeding unity, while
shot noise displaces the reconstructed Bob state by at most
$\sim\!0.066$ at $3\sigma$, so no plausible fluctuation can carry
$\rho_B$ outside the polytope.

Faithfulness holds for all four $\varphi$ values at $\geq 8\sigma$,
with every measured $C(\rho_C)$ within $0.025$ of the analytic value
of Lemma~\ref{lem:C_formula}.
Depolarising noise alone accounts for the residual deficit. Writing
$\rho_C = \eta\,\proj{\varphi} + (1-\eta)I/2$ with $\eta = 2F-1$ gives
$C_{\rm dep} = \eta\,C(\varphi)$, and all four measured values sit
within $1.1\sigma$ of this prediction, the largest discrepancy being at
$\varphi = \pi/8$. We see no sign of a $\varphi$-dependent systematic
on top of it. Tomographic reconstruction points the same way: the delivered
states show no significant phase-calibration offset (mean
$-1.4^\circ$, spread $1.5^\circ$). Such an offset would be a coherent
rather than a stochastic error, and because $C(\varphi)$ is sharply
peaked it would cost disproportionately more magic content than
fidelity, so it is worth checking for separately.
For 15-to-1 distillation~\cite{Bravyi2005} the relevant quantity is the
fidelity with the nearest $T$-type axis after twirling,
$F_T=\tfrac12(1+\mathbf{r}\cdot\hat{\mathbf{t}})$, rather than the fidelity with
each state's own target. Under the depolarising model above,
$F_T=\tfrac12[1+(2F-1)\cos\delta]$ with $\delta$ the angle from $\varphi$ to the
nearest $T$-type axis, giving $F_T=0.937$, $0.959$, $0.949$ and $0.959$ for
$\varphi=\pi/8$, $\pi/4$, $\pi/3$ and $3\pi/4$. All exceed the 15-to-1 threshold
$F^*\approx0.8585$ (the fixed point of its exact output-error map), so the delivered
states are of distillable quality.
We had earlier run the same protocol on a different qubit assignment.
Those results reproduce the ones reported here: $C(\rho_B)$ fell below
solver tolerance in every case, and the reconstructed $C(\rho_C)$
values agree with Table~\ref{tab:hw} to within $0.026$ in magic
content --- a deviation of $1.2\sigma$ once the bootstrap uncertainties
of the two runs are combined in quadrature. The
security guarantee thus holds in all eight runs across the two
assignments.

\section{Discussion}
\label{sec:discussion}

\subsection{MSS vs standard QSS}

Standard QSS~\cite{Hillery1999,Cleve1999} prevents unauthorised parties
from learning the identity of the secret state, and Corollary~\ref{cor:coalition}
is exactly this guarantee for MSS. The GHZ code realises it perfectly for
equatorial states (Remark~\ref{rem:equator}), which is where the $T$-type magic
lies. MSS adds a resource-theoretic guarantee: a single unauthorised share is
$I/2$, which carries no magic under any unitary.
In multi-server BQC the relevant threat is unilateral computational
universality rather than information leakage: there the angle $\varphi=\pi/4$
is public, and what is protected is the capability, not the key.

\subsection{Note on C as a magic measure}

$C$ is invariant under single-qubit Cliffords.
For multi-qubit states, $C$ can increase under certain free operations
(e.g., $H\otimes I$)~\cite{Dutta2026},
reflecting the obstruction to a Clifford-covariant qubit phase space
identified in~\cite{Raussendorf2017}.
This does not affect the present results: all critical quantities
($C(\rho_B)$, $C(\rho_C)$, $C(\varphi)$) are single-qubit values.
The security proof uses only $C(I/2) = 0$ and the $I/2$ fixed-point
argument, neither of which requires $C$ to be a monotone
under the full Clifford group~\cite{Dutta2026}.
The robustness of magic $\Gamma$~\cite{Howard2017}, which is a monotone, gives
the same picture: a CNOT ladder maps $P(\varphi)\ket{+}\otimes\ket{0}^{\otimes(n-1)}$
to the $n$-party GHZ state with phase $\varphi$, and $\Gamma$ is invariant under
Clifford unitaries and under adding or discarding stabilizer ancillas, so
$\Gamma=|\sin\varphi|+|\cos\varphi|$ for every $n$. The distributed state carries
exactly one qubit's worth of magic, and the delivery loses none of it.

\subsection{Open problems}

The $(k,n)$ generalisation with $k < n$ faces two distinct
obstacles. Structurally, every proper subset of GHZ shares is independent of
$\varphi$, so no group smaller than all $n$ parties can deliver
$P(\varphi)\ket{+}$ within this construction; a different code
(graph states are natural candidates) is required.
Fundamentally, since the recipient obtains a pure quantum state, the
quantum secret sharing bound of Cleve, Gottesman, and
Lo~\cite{Cleve1999} applies: $k > n/2$ is necessary, ruling out
$k \leq n/2$ regardless of the entanglement resource employed.
A one-sided device-independent version, in which steering is certified
from the full assemblage with untrusted coalition devices (for instance with
test rounds in which the dealer announces her phase), would combine
Remark~\ref{rem:certification}(a) with the witness of
Proposition~\ref{prop:1SDI}.
The present adversary model also does not cover integrity: a party that
announces a false outcome leaves the recipient with $ZP(\varphi)\ket{+}$, which has
the same magic but implements the wrong gate. Verification rounds that detect
such deviations are a natural next step.
The timescale on which delivered magic survives idling, characterised
in the companion paper~\cite{DuttaTushar2026b}, could be mapped across
platforms to enable hardware-adapted protocols.
Whether the two-resource cooperation identified here is a general
mechanism for noise-resilient quantum resources is open,
with natural connections to magic-preserving quantum error correction
via the Heisenberg picture of $C$~\cite{Dutta2026}.

\section*{Acknowledgments}

We thank Dr.\ Chandan Datta (IISER Kolkata) for valuable
insights on the application part, Dr.\ Albert Rico (University of
Siegen, Germany) for insightful discussions on threshold structures for
quantum resources and for observations that sharpened the coalition
security analysis, and Cameron Foreman (Quantinuum) for careful reading
and questions that improved the presentation of the adversary model.
IBM Quantum access was provided through the Open Plan.

\end{document}